\documentclass[12pt]{article}
\usepackage[usenames]{color}
\usepackage{xcolor}
\usepackage{graphicx, subfigure}
\usepackage{amsmath, amsthm, amssymb}
\usepackage{amsfonts}
\usepackage{fullpage}
\usepackage{ifthen}
\usepackage{url}
\usepackage{hyperref}
\usepackage[sort&compress]{natbib}
\usepackage{caption}
\usepackage{multirow}
\usepackage{bm}
\usepackage{tikz-qtree}

\usepackage[linesnumbered,algoruled,boxed,lined,commentsnumbered]{algorithm2e}

\newtheorem{theorem}{Theorem}[section]

\makeatletter
\def\@biblabel#1{}
\makeatother

\theoremstyle{plain}

\theoremstyle{definition}

\theoremstyle{remark}

\newcommand{\E}{\mathsf{E}}
\newcommand{\V}{\mathsf{V}}

\title{Estimating a function of the scale parameter in a gamma distribution with bounded variance} \author{ Jun Hu\footnote{Jun Hu is an Assistant Professor in the Department of Mathematics and Statistics, Oakland University, 146 Library Drive, Rochester, MI 48309, USA. Tel.:~248-370-3434. Email address: junhu@oakland.edu.} \; Ibtihal Alanazi\footnote{Ibtihal Alanazi is a Ph.D. student in the Department of Mathematics and Statistics, Oakland University, 146 Library Drive, Rochester, MI 48309, USA. Email address: ialanazi@oakland.edu.} \; Zhe Wang\footnote{Zhe Wang is an Assistant Professor in the Data Analytics Program at Denison University, 100 West College Street, Granville, Ohio 43023, USA. Email address: wangz@denison.edu.} }
\date{}

\begin{document}

\maketitle

\begin{abstract}

Given a gamma population with a known shape parameter $\alpha$, we develop a general theory for estimating a function $g(\cdot)$ of the scale parameter $\beta$ with bounded variance. We begin by defining a sequential sampling procedure with $g(\cdot)$ satisfying some desired condition in proposing the stopping rule, and show the procedure enjoys appealing asymptotic properties. After these general considerations, we substitute $g(\cdot)$ with specific functions including the gamma mean, the gamma variance, the gamma rate parameter, and a gamma survival probability as four possible illustrations. For each illustration, Monte Carlo simulations are carried out to justify the remarkable performance of our proposed sequential procedure. This is further substantiated by two real data studies on the survival times of female dementia patients and on the number of days that the seeds of marigold need to flower, respectively.    

\smallskip

\end{abstract}

\setcounter{section}{0}
\setcounter{equation}{0}
\section{Introduction}

The gamma distribution, as one of the most commonly used probability distributions for modeling data, plays an important role in many areas such as business management, environmental science, reliability engineering, and medical science. For example, \cite{Nenes et al. (2010)} considered modeling the demand for commercial items employing a gamma distribution with a probability mass at zero; \cite{Rohan et al. (2015)} investigated the use of the gamma distribution to determine the half-life of rotenone applied in freshwater; \cite{Ibrahim et al. (2019)} applied the gamma distribution to evaluate and predict the reliability and the projected lifetime of phosphor-converted white LEDs based on luminous flux degradation; and \cite{Cui et al. (2023)} proposed a three-layer feature selection combined with a gamma distribution-based generalized linear model for anticancer drug response prediction.   

It is safe to say that there is a large volume of articles on the estimation of the gamma distribution parameters, most of which are based on fixed-sample-size procedures. In a lot of statistical inference problems, however, requirements on some predetermined accuracy (e.g., confidence interval estimation with a fixed width, point estimation with a bounded risk, testing hypotheses with controlled type-I and type-II errors, etc.) leads to the nonexistence of such fixed-sample-size procedures. As a consequence, sequential or multistage sampling procedures are necessary. In the context of gamma parameters estimation, \cite{Takada and Nagata (1995)} constructed the fixed-width sequential confidence interval for the mean of a gamma population; \cite{Isogai and Uno (1995)} considered the minimum risk point estimation problem for a gamma mean, and proposed a sequential sampling procedure shown to be asymptotically better than the one given by \cite{Woodroofe (1977)}; \cite{Zacks and Khan (2011)} developed both two-stage and purely sequential estimation procedures for fixed-width confidence interval of the scale parameter of a gamma distribution when the shape parameter is known; \cite{Mahmoudi and Roughani (2015)} and \cite{Roughani and Mahmoudi (2015)} subsequently investigated the bounded risk estimation problem for the gamma scale parameter in a two-stage sampling procedure; \cite{Mahmoudi et al. (2019)} further studied the same problem in a purely sequential sampling procedure; \cite{Zhuang et al. (2020)} established sequential fixed-accuracy confidence intervals for the survival function of a gamma distribution; and \cite{Bapat (2022)} obtained novel two-stage sampling procedures to estimate the ratio and the sum of shape parameters coming from two independent gamma populations.     

In a parallel direction, there is also interest in estimating a function of the parameter(s) sequentially to make the inference problem more practical. To list a few, \cite{Uno et al. (2004)} put forward sequential point estimation of a function of the exponential scale parameter subject to squared error loss plus a linear cost; \cite{Mukhopadhyay and Wang (2019)} developed a general theory of sequential minimum risk point estimation of a function of a normal mean; \cite{Banerjee and Mukhopadhyay (2021)} designed multistage minimum risk point estimation strategies for a function of a normal mean; \cite{Mukhopadhyay and Li (2022)} proposed a purely sequential minimum risk point estimation procedure for a survival function in an exponential distribution; and most recently, \cite{Mukhopadhyay and Li (2024)} formulated a general multistage sampling methodology for estimating a function of an unknown gamma scale parameter $\beta$ with a known shape parameter $\alpha$ under a squared error loss function plus the cost of sampling. For a broad-ranging review of sequential analysis as a powerful tool, one may refer to many resources including the following monographs: Anscombe(\citeyear{Anscombe (1952)}, \citeyear{Anscombe (1953)}), \cite{Robbins (1959)}, \cite{Chow and Robbins (1965)}, \cite{Siegmund (1985)}, \cite{Ghosh et al. (1997)}, and \cite{Mukhopadhyay and de Silva (2009)}.   

In this paper, we examine a sequential approach to tackling the so-called \textit{bounded variance point estimation} (BVPE) problem for a function $g(\cdot)$ of the scale parameter $\beta$ of a gamma distribution, assuming the shape parameter $\alpha$ is known to us. That is, the sequential methodology is devised to constrain the variance of the estimator to a prescribed bound $b(>0)$. BVPE can be applied in many fields such as actuarial science and medical science, showing its practical applicability. One may see \cite{Hu and Hong (2022)} and \cite{Hu and Zhuang (2022)} for more details. 

The remainder of the paper is organized as follows. In Section \ref{Proc}, we lay down some useful preliminaries, establish a general theory, and demonstrate the appealing properties of sequential BVPE for a function of the gamma scale parameter. In Section \ref{Illu}, we exhibit four illustrations where the general framework developed in Section \ref{Proc} is readily applied, namely, the gamma mean, the gamma variance, the gamma rate parameter, and the survival probability. Extensive sets of Monte Carlo simulations are also carried out to verify the theoretical findings empirically. Section \ref{RDA} includes real data analysis to showcase the possibility of incorporating our procedure in real life problems. The paper is wrapped up with some brief overall thoughts in Section \ref{Conc}.

\setcounter{equation}{0}
\section{Preliminaries, sequential estimation, and properties}\label{Proc}

Suppose that $X_1,X_2,...$ are independent and identically distributed (i.i.d.) observations from a gamma population $\Gamma(\alpha,\beta)$ with the associated probability density function (pdf) given by
$$f(x;\alpha,\beta) = \frac{1}{\Gamma(\alpha)\beta^{\alpha}}x^{\alpha-1}e^{-x/\beta}, \ \text{ for } x>0.$$
Here, the scale parameter $\beta(>0)$ is of interest and remains unknown, while the shape parameter $\alpha(>0)$ is \textit{known} to us. We desire to estimate $g(\beta)$, a \textit{strictly monotonic} function of the scale parameter. Since $\beta$ takes only positive values, we let the function $g(\cdot)$ be defined on the positive real line for simplicity; and assume that for all $x \in \mathcal{R}^{+}$, $g(x)$ is twice continuously differentiable. Denote the first two derivatives by $g^{\prime}(x)$ and $g^{\prime\prime}(x)$, respectively. In the spirits of \cite{Mukhopadhyay and Wang (2019)} and \cite{Mukhopadhyay (2021)}, we adopt a similar but relatively relaxed condition as follows:
\begin{equation}\label{C1}
\left|g^{\prime\prime}(x)\right| \le \sum_{j=1}^{d}a_{j}x^{k_{j}},
\end{equation}
where $d(\ge1)$ is a fixed integer, $a_{j}(\ge0)$ and $k_{j}$ (not necessarily nonnegative) are appropriate fixed real numbers, for $j=1,...,d$. One may wonder if the condition \eqref{C1} is too strong to be practical. In fact, there are a lot of common and meaningful functions of $\beta$ regarding a gamma distribution that satisfy this condition. Here are some typical examples.
\begin{enumerate}
	\item[(i)] The mean of a gamma distribution $g(\beta)=\alpha\beta$:
	\begin{equation*}
	\begin{split}
	g^{\prime}(\beta) = \alpha>0, \ |g^{\prime\prime}(\beta)| = 0,
	\end{split}
	\end{equation*}
	where $d=1,a_{1}=0$, and $k_{1}=0$.
	\item[(ii)] The variance of a gamma distribution $g(\beta)=\alpha\beta^2$:
	\begin{equation*}
	\begin{split}
	g^{\prime}(\beta) = 2\alpha\beta>0, \ |g^{\prime\prime}(\beta)| = 2\alpha,
	\end{split}
	\end{equation*}
	where $d=1,a_{1}=2\alpha$, and $k_{1}=0$.
	\item[(iii)] The rate parameter of a gamma distribution $g(\beta)=\beta^{-1}$:
	\begin{equation*}
	\begin{split}
	g^{\prime}(\beta) = -\beta^{-2}<0, \ |g^{\prime\prime}(\beta)| = 2\beta^{-3},
	\end{split}
	\end{equation*}
	where $d=1,a_{1}=2$, and $k_{1}=-3$.
	\item[(iv)] A survival probability of a gamma distribution $g(\beta)=\Pr(X>c)=\int_{c/\beta}^{\infty}\frac{1}{\Gamma(\alpha)}x^{\alpha-1}e^{-x}dx$ with $c>0$ being an appropriate constant:
	\begin{equation*}
	\begin{split}
	g^{\prime}(\beta) &= -\frac{c^{\alpha}}{\Gamma(\alpha)}\beta^{-(\alpha+1)}e^{-c/\beta}<0,\\ 
	g^{\prime\prime}(\beta) &= \frac{c^{\alpha}}{\Gamma(\alpha)}\beta^{-(\alpha+2)}e^{-c/\beta}(\alpha+1+c\beta^{-1})\\
	&\Rightarrow |g^{\prime\prime}(\beta)|\le \frac{c^{\alpha}(\alpha+1)}{\Gamma(\alpha)}\beta^{-(\alpha+2)}+\frac{c^{\alpha+1}}{\Gamma(\alpha)}\beta^{-(\alpha+3)},
	\end{split}
	\end{equation*}
	where $d=2,a_{1}=\frac{c^{\alpha}(\alpha+1)}{\Gamma(\alpha)},k_{1}=-(\alpha+2),a_{2}=\frac{c^{\alpha+1}}{\Gamma(\alpha)}$, and $k_{2}=-(\alpha+3)$.
\end{enumerate}

Having recorded $X_1,...,X_n,n\ge1$, a maximum likelihood estimator (MLE) of $\beta$ is given by $\hat{\beta}_n = \alpha^{-1}\bar{X}_n$, where $\bar{X}_n=n^{-1}\sum_{i=1}^{n}X_i$ denotes the sample mean. According to the invariant property, a natural MLE of $g(\beta)$ is given by $g(\hat{\beta}_n)$. As for its variance denoted by $\V[g(\hat{\beta}_n)]$, in general, it may not be easy to find an explicit expression. Hence, we employ the delta method to obtain an approximation, which is valid as $g(\cdot)$ has a continuous second derivative. Applying Taylor's theorem, we have 
\begin{equation*}
g(\hat{\beta}_n)-g(\beta) = g^{\prime}(\beta)(\hat{\beta}_n-\beta)+\frac{1}{2}g^{\prime\prime}(\xi_n)(\hat{\beta}_n-\beta)^2,
\end{equation*}
where $\xi_n$ is a random variable lying between $\hat{\beta}_n$ and $\beta$. Then, 
\begin{equation}\label{delta}
\V[g(\hat{\beta}_n)] = \frac{\{g^{\prime}(\beta)\}^2\beta^2}{\alpha n} + \frac{1}{4}\V[g^{\prime\prime}(\xi_n)(\hat{\beta}_n-\beta)^2] + g^{\prime}(\beta)\textsf{Cov}(\hat{\beta}_n-\beta,g^{\prime\prime}(\xi_n)(\hat{\beta}_n-\beta)^2).
\end{equation}
Under the condition \eqref{C1},
\begin{equation*}
\begin{split}
|{g^{\prime\prime}(\xi_n)}| &\le \sum_{j=1}^{d}a_{j}\xi_n^{k_{j}} \le \sum_{j=1}^{d}a_{j}(\alpha^{-1}\bar{X}_n+\beta)^{k_{j}}.
\end{split}
\end{equation*}
As $\bar{X}_n \sim \Gamma(n\alpha,\beta/n)$, it is clear that $\bar{X}_n^{pk_{j}}$ is integrable given $pk_{j}+n\alpha>0$. Therefore, $\{g^{\prime\prime}(\xi_n)\}^p$ is integrable when 
\begin{equation}\label{momcond}
p\min_{j\in\{1,...,d\}}k_{j}+n\alpha>0.
\end{equation}
Suppose there exists a $p=2p_0>2$ satisfying \eqref{momcond}. Then, by H{\"o}lder's inequality, 
\begin{equation}\label{delta1}
\begin{split}
\V[g^{\prime\prime}(\xi_n)(\hat{\beta}_n-\beta)^2] &\le \E^{1/p_0}[\{g^{\prime\prime}(\xi_n)\}^{2p_0}]\cdot \E^{1/q_0}[(\hat{\beta}_n-\beta)^{4q_0}]\\
&= O(1)O(n^{-2}) = O(n^{-2}),
\end{split}
\end{equation}
where $1/p_0+1/q_0=1$. In a similar fashion,
\begin{equation}\label{delta2}
\begin{split}
\textsf{Cov}(\hat{\beta}_n-\beta,g^{\prime\prime}(\xi_n)(\hat{\beta}_n-\beta)^2) &\le \V^{1/2}[\hat{\beta}_n-\beta] \cdot \V^{1/2}[g^{\prime\prime}(\xi_n)(\hat{\beta}_n-\beta)^2]\\
&= O(n^{-1/2})O(n^{-1}) = O(n^{-3/2}),
\end{split}
\end{equation}
Putting together \eqref{delta}, \eqref{delta1} and \eqref{delta2}, we arrive at
\begin{equation}\label{Vapprox}
\V[g(\hat{\beta}_n)] = \frac{\{g^{\prime}(\beta)\}^2\beta^2}{\alpha n} + O(n^{-3/2}).
\end{equation}
This suggests that the delta method gives a good approximation of $\V[g(\hat{\beta}_n)]$ with an error up to the order $O(n^{-3/2})$, assuming the condition \eqref{C1}, and that \eqref{momcond} holds for some $p>2$. Now, our goal is to control the variance of the estimator $g(\hat{\beta}_n)$ under a predetermined small level $b(>0)$, that is,
\begin{equation}\label{bv}
\V[g(\hat{\beta}_n)] \le b.
\end{equation} 
Omitting the small quantity $O(n^{-3/2})$ in \eqref{Vapprox}, we have that 
\begin{equation}\label{oss}
\V[g(\hat{\beta}_n)] \approx \frac{\{g^{\prime}(\beta)\}^2\beta^2}{\alpha n} \le b \Rightarrow n \ge \frac{\{g^{\prime}(\beta)\}^2\beta^2}{\alpha b} = n^{*}, \text{ say.}
\end{equation}
Here, $n^*$ is defined as the \textit{optimal sample size} whose magnitude remains unknown, though, since it involves the unknown scale parameter $\beta$. A sequential sampling procedure can be implemented to solve this estimation problem, where one estimates $\beta$ by updating its MLE at every stage as needed.

Beginning with a pilot sample of size $m(\ge1)$, $X_1,...,X_m$, we propose the following sequential procedure which leads to a final sample size required for bounding the variance of the estimator. The associated stopping rule is given by
\begin{equation}\label{SR}
N = N(b) = \inf\left\{ n\ge m: n \ge \frac{\{g^{\prime}(\hat{\beta}_n)\}^2\hat{\beta}_n^2}{\alpha b} \right\},
\end{equation}  
where $\hat{\beta}_n = \bar{X}_n/\alpha$. Concerning the implementation of \eqref{SR}, if $m \ge \frac{\{g^{\prime}(\hat{\beta}_m)\}^2\hat{\beta}_m^2}{\alpha b}$
is already satisfied, we do not take any additional observations, and the final sample size is $N = m$. Otherwise, we collect one more observation $X_{m+1}$ and update $\hat{\beta}_{m+1}$ to check with the stopping rule \eqref{SR}. Sampling is then terminated at the first time $N = n(\ge m)$ such that $n \ge \frac{\{g^{\prime}(\hat{\beta}_n)\}^2\hat{\beta}_n^2}{\alpha b}$ occurs. At last, with the fully accrued data ${X_1, ..., X_m, ..., X_N}$, we construct the bounded variance
point estimator for $g(\beta)$ by
\begin{equation}\label{BVPE}
g(\hat{\beta}_N) = g(\bar{X}_N/\alpha).
\end{equation}
Obviously, we can claim that $\Pr\{N<\infty\}=1$ and $N \to \infty$ with probability 1 (w.p.1) as $b\to0$.

The sequential estimation procedure \eqref{SR} is efficient in terms of the final sample size required, as its expected value will be close to the optimal sample size $n^*$ defined in \eqref{oss} in some certain manner. In, we conclude the so-called asymptotic first-order efficiency in Theorem \ref{Thm1} generally, and second-order analysis may be available according to the specific function of $g(\cdot)$ under consideration. A few typical illustrations will be provided in Section \ref{Illu}.   

\begin{theorem}\label{Thm1}
Assume the condition \eqref{C1}, and that $p\min_{j\in\{1,...,d\}}k_{j}+m\alpha>0$ for some $p>2$. For the sequential estimation procedure \eqref{SR}, we have that as $b \to 0$,
\begin{equation}\label{eff}
\E[N/n^*] = 1 + o(1).
\end{equation}
\end{theorem}
\begin{proof}
For $n=m,m+1,...$, define 
$$y_n = \left[\frac{g^{\prime}(\hat{\beta}_n)\hat{\beta}_n}{g^{\prime}(\beta)\beta}\right]^2.$$
Since $\hat{\beta}_n$ is the MLE of $\beta$ and $g(\cdot)$ is a continuous function, then $\lim_{n\to\infty}y_n=1$ w.p.1. Let $f(n)=n$ so it is trivial that $\lim_{n\to\infty}f(n)=\infty$ and $\lim_{n\to\infty}f(n)/f(n-1)=\infty$. The stopping rule \eqref{SR} can be rewritten as
$$N = \inf\{ n \ge m: y_n \le f(n)/n^* \},$$
which matches that of \cite{Chow and Robbins (1965)}. Hence, $\lim_{n\to\infty}n^{*-1}N=1$ w.p.1 by their Lemma 1. It suffices to show the uniform integrability of $\{n^{*-1}N:0<b\le1\}$ along the lines of \cite{Ghosh and Mukhopadhyay (1980)} and \cite{Hu and Mukhopadhyay (2019)}. 

Write for any $c>0$,
\begin{equation}\label{t1}
\E\left[ \{N/n^*\}\mathbb{I}(N>cn^*) \right] = c\Pr(N>cn^*) + \int_{c}^{\infty}\Pr(N>xn^*)dx,
\end{equation}
where $\mathbb{I}(A)$ represents the indicator function of an event $A$. Let $t = \lfloor xn^* \rfloor$, where $\lfloor u \rfloor$ indicates the largest integer that is strictly smaller than $u(>0)$. Select $b_0\le1$ such that $c > c_0(\ge 2)$ for small $b \le b_0$ and $g^{\prime}(\hat{\beta}_{n})/g^{\prime}(\beta)<2$ w.p.1 for $n \ge \lfloor cn^* \rfloor$. Then,
\begin{equation}
\begin{split}
\Pr(N>xn^*) & \le \Pr(N>t) \le \Pr\left( \left[\frac{g^{\prime}(\hat{\beta}_t)\hat{\beta}_t}{g^{\prime}(\beta)\beta}\right]^2 \ge t \right)\\
& \le \Pr\left( |\hat{\beta}_t-\beta| \ge (t^{1/2}/2-1)\beta \right)\\
& \le \frac{\E[|\hat{\beta}_t-\beta|^{s}]}{[(t^{1/2}/2-1)\beta]^{s}}=O(x^{-s}), 
\end{split}
\end{equation}
for some appropriate $s>1$. It leads to $c\Pr(N>cn^*)=O(c^{1-s})$ and $\int_{c}^{\infty}\Pr(N>xn^*)dx=O(c^{1-s})$ for sufficiently large $c$. From \eqref{t1}, we can claim that $\{n^{*-1}N:0<b\le1\}$ is uniformly integrable, and therefore \eqref{eff} stands.
\end{proof}

As our ultimate target is to bound the variance of the point estimator of $g(\beta)$, it is anticipated that $\V[g(\hat{\beta}_N)]$ remains below the prescribed small level $b>0$ or approximately $b$. The next theorem demonstrates that the sequential estimation procedure \eqref{SR} yields an estimator whose variance is close to $b$ in some certain way.

\begin{theorem}\label{Thm2}
Assume the condition \eqref{C1}, and that $p\min_{j\in\{1,...,d\}}k_{j}+m\alpha>0$ for some $p>2$. For the sequential estimation procedure \eqref{SR}, we have that as $b \to 0$,
\begin{equation}\label{avar}
\V[g(\hat{\beta}_N)] = b + O(b^{3/2}).
\end{equation}
\end{theorem}
\begin{proof}
We set out to prove \eqref{avar} by appealing to \cite{Yu (1989)}. Note that for a random sample of the gamma population $\Gamma(\alpha,\beta)$, $X_1,...,X_n$, the Fisher information about $\beta$ is given by
$$I_{\beta,n}(\beta) = \E\left[ -\frac{d^2}{d\beta^2}\log\prod_{i}^{n}f(X_i;\beta) \right] = \frac{n\alpha}{\beta^2}.$$
Using $\lambda=g(\beta)$ to reparameterize the gamma population, the corresponding Fisher information with respect to $\lambda$ is
$$I_{\lambda,n}(\lambda) = I_{\beta,n}(\beta) \left(\frac{d \beta}{d \lambda}\right)^2 = \frac{n\alpha}{\{g^{\prime}(\beta)\}^{-2}\beta^2}.$$ 
Since the magnitude of $\beta$ is veiled, we use the observed Fisher information at $\hat{\beta}_n$ to estimate $I_{\lambda,n}(\lambda)$. Then, one has 
$$I_{\hat{\lambda}_n,n}(\hat{\lambda}_n) = \frac{n\alpha}{\{g^{\prime}(\hat{\beta}_n)\}^{-2}\hat{\beta}_n^2}.$$
The stopping rule \eqref{SR} can be rewritten as
\begin{equation}\label{SR2}
N = \inf\{ n \ge m: I_{\hat{\lambda}_n,n}(\hat{\lambda}_n) \le b \},
\end{equation}     
which matches the stopping rule proposed in \cite{Yu (1989)}. Therefore, one can immediately obtain that as $b \to 0$,
\begin{equation}\label{anorm1}
\sqrt{N}(g(\hat{\beta}_N)-g(\beta)) \overset{d}{\longrightarrow} N(0,\{g^{\prime}(\beta)\}^2\beta^2/\alpha),
\end{equation} 
and 
\begin{equation}\label{anorm2}
\sqrt{N}(\hat{\beta}_N-\beta) \overset{d}{\longrightarrow} N(0,\beta^2/\alpha).
\end{equation} 
Observe that for any $q \ge 1$, $\hat{\beta}_N^q \le \sup_{n\ge m} \hat{\beta}_n^q$, and $\E(\hat{\beta}_n^q) \sim \beta^{q} < \infty$. The uniform integrability of $\{ \hat{\beta}_N^q, 0< b \le 1 \}$ can be established through Wiener's (\citeyear{Wiener (1939)}) ergodic theorem. Apply Taylor's theorem, and we express
\begin{equation}
g(\hat{\beta}_N)-g(\beta) = g^{\prime}(\beta)(\hat{\beta}_N-\beta)+\frac{1}{2}g^{\prime\prime}(\kappa_N)(\hat{\beta}_N-\beta)^2,
\end{equation}   
where $\kappa_N$ is a random variable lying between $\hat{\beta}_N$ and $\beta$. Then, following our previously used arguments leading to \eqref{Vapprox}, we can claim that
$$\V[g(\hat{\beta}_N)] = b + O(b^{3/2}).$$ The proof is now complete.  
\end{proof}

\setcounter{equation}{0}
\section{Illustrations}\label{Illu}

In this section, we substitute the general $g(\beta)$ with the aforementioned practically useful functions for illustrative purposes. For each illustration, we specify the stopping rule and carry out Monte Carlo simulations to validate the theoretical findings empirically.

\subsection{Illustration 1: the gamma mean}

To estimate the gamma mean $g(\beta)=\alpha\beta$, we follow the general framework of sequential estimation established in Section \ref{Proc}, and propose the procedure ${\mathcal{P}}_1$ associated with the stopping rule given by
\begin{equation}\label{P1}
N_1 = \inf\left\{ n \ge m: n \ge \frac{\alpha\hat{\beta}_n^2}{b} \right\}.
\end{equation}

As was pointed out, $d=1,a_1=\alpha$ and $k_1=0$ in this case so that any pilot sample size $m\ge1$ will satisfy the condition \eqref{momcond}. Next, note that the stopping rule can be alternatively represented by
\begin{equation}\label{P11}
N_1 = \inf\left\{ n \ge m: \sum_{i=1}^{n}W_i \le (\alpha b)^{1/2}\beta^{-1}n^{3/2} \right\},
\end{equation}
where $W_1,...,W_n$ are independent and identically distributed (i.i.d.) $\Gamma(\alpha,1)$ random variables. It then matches the stopping rule of \cite{Woodroofe (1977)}. Therefore, in addition to the asymptotic first-order efficiency as per \eqref{eff}, we have the following second-order approximation by appealing to the nonlinear renewal theory: for $m>4/\alpha$,
\begin{equation}\label{eff1}
\E[N_1-n_1^*] = \frac{1}{2}-\frac{1}{\alpha}-\frac{2}{\alpha}\sum_{n=1}^{\infty}\frac{1}{n}\E\left( \sum_{i=1}^{n}W_i-\frac{3}{2}\alpha n \right)^++o(1)
\end{equation} 
as $b \to 0$, where $n_1^*=\alpha\beta^2/b$ is the corresponding optimal sample size. One may refer to Lai and Siegmund (\citeyear{Lai and Siegmund (1977)}, \citeyear{Lai and Siegmund (1979)}) and \cite{Mukhopadhyay and de Silva (2009)} for more background details.

An extensive set of Monte Carlo simulations have been conducted to investigate the performance of the sequential estimation procedure $\mathcal{P}_1$. We generated pseudorandom samples from a gamma population with $\alpha=2$ and $\beta=2$, but pretended that $\beta$ was unknown. In this case, the second-order expansion in \eqref{eff1} is computable, and a separate R program was written to approximate it numerically. In the spirits of \citet[Table 3.8.1]{Mukhopadhyay and Solanky (1994)}, we excluded any term smaller than $10^{-15}$ in magnitude in the infinite sum, and obtained that $\E[N_1-n_1^*] \approx -0.5680$. 

Then, we fixed the pilot sample size $m=20$. A wide range of values of $b$ including $0.160,0.080,0.040,0.016,0.008$ and $0.004$ were considered, so that $n^*_1$ turned out 50, 100 (small), 200, 500 (moderate), and 1000, 2000 (large) accordingly. The procedure was repeated for a total of 10,000 times. In Table \ref{Tab. 1}, we record the average final sample size $\bar{n}_1$ with the associated standard deviation $s(n_1)$, the first-order approximation $\bar{n}_1/n_1^*$, the second-order approximation $\bar{n}_1-n^*_1$, the variance of the estimates $\V[g(\hat{\beta})_{n_1}]$, and the ratio of $\V[g(\hat{\beta}_{n_1})]$ to $b$.    

\begin{table}[h!] 
\small
\captionsetup{font=small}
\caption{Simulations from $\Gamma(2,2)$ with $m=20$ under $10,000$ runs implementing the sequential estimation procedure $\mathcal{P}_1$}
\label{Tab. 1}\par
\vskip .2cm
\centerline{\tabcolsep=3truept\begin{tabular}{cccccccc} \hline 
$b$ & $n_1^*$ & $\bar{n}_1$ & $s(n_1)$ & $\bar{n}_1/n_1^*$ & $\bar{n}_1-n_1^*$ & $\V(g(\hat{\beta}_{n_1}))$ & $\V(g(\hat{\beta}_{n_1}))/b$ \\ \hline
\multicolumn{1}{r}{$0.160$} & \multicolumn{1}{r}{$50$} & \multicolumn{1}{r}{$49.2343$} & \multicolumn{1}{r}{$10.3368$} & \multicolumn{1}{c}{$0.9847$} & \multicolumn{1}{c}{$-0.7657$} & \multicolumn{1}{c}{$0.187391$} & \multicolumn{1}{c}{$1.171194$}\\
\multicolumn{1}{r}{$0.080$} & \multicolumn{1}{r}{$100$} & \multicolumn{1}{r}{$99.4338$} & \multicolumn{1}{r}{$14.4372$} & \multicolumn{1}{c}{$0.9943$} & \multicolumn{1}{c}{$-0.5662$} & \multicolumn{1}{c}{$0.087127$} & \multicolumn{1}{c}{$1.089093$}\\
\multicolumn{1}{r}{$0.040$} & \multicolumn{1}{r}{$200$} & \multicolumn{1}{r}{$199.6266$} & \multicolumn{1}{r}{$20.1623$} & \multicolumn{1}{c}{$0.9981$} & \multicolumn{1}{c}{$-0.3734$} & \multicolumn{1}{c}{$0.041503$} & \multicolumn{1}{c}{$1.037576$}\\
\multicolumn{1}{r}{$0.016$} & \multicolumn{1}{r}{$500$} & \multicolumn{1}{r}{$499.3923$} & \multicolumn{1}{r}{$31.8152$} & \multicolumn{1}{c}{$0.9988$} & \multicolumn{1}{c}{$-0.6077$} & \multicolumn{1}{c}{$0.016356$} & \multicolumn{1}{c}{$1.022236$}\\
\multicolumn{1}{r}{$0.008$} & \multicolumn{1}{r}{$1000$} & \multicolumn{1}{r}{$999.0571$} & \multicolumn{1}{r}{$44.8184$} & \multicolumn{1}{c}{$0.9991$} & \multicolumn{1}{c}{$-0.9429$} & \multicolumn{1}{c}{$0.008074$} & \multicolumn{1}{c}{$1.009238$}\\
\multicolumn{1}{r}{$0.004$} & \multicolumn{1}{r}{$2000$} & \multicolumn{1}{r}{$1999.6929$} & \multicolumn{1}{r}{$62.8701$} & \multicolumn{1}{c}{$0.9998$} & \multicolumn{1}{c}{$-0.3071$} & \multicolumn{1}{c}{$0.003958$} & \multicolumn{1}{c}{$0.989538$}\\
\hline
\end{tabular}}
\end{table}

From Table \ref{Tab. 1}, it is clear that as $b(n^*_1)$ gets smaller (larger), the first-order efficiency term $\bar{n}_1/n_1^*$ approaches 1, and the second-order efficiency term $\bar{n}_1-n_1^*$ is comparable to $-0.5680$. Across the board, the variance estimate $\V(g(\hat{\beta}_{n_1}))$ is close to the target $b$, and the ratio $\V(g(\hat{\beta}_{n_1}))/b$ is comparable to 1. These results show that the developed sequential estimation procedure $\mathcal{P}_1$ has performed remarkably well. 

\subsection{Illustration 2: the gamma variance}

In a parallel fashion, we propose the sequential procedure $\mathcal{P}_2$ as follows to estimate the gamma variance $g(\beta)=\alpha\beta^2$.
\begin{equation}\label{P2}
N_2 = \inf\left\{ n \ge m: n \ge \frac{4\alpha\hat{\beta}_n^4}{b} \right\}.
\end{equation}

As was pointed out, $d=1,a_1=2\alpha$ and $k_1=1$ in this instance so that any pilot sample size $m\ge1$ will satisfy the condition \eqref{momcond}. Similarly, a rewrite of the stopping rule \eqref{P2} will bring it in line with that of \cite{Woodroofe (1977)}.
\begin{equation}\label{P21}
N_1 = \inf\left\{ n \ge m: \sum_{i=1}^{n}W_i \le (\alpha^3 b/4)^{1/4}\beta^{-1}n^{5/4} \right\},
\end{equation}
where $W_1,...,W_n$ are i.i.d. $\Gamma(\alpha,1)$ random variables. Therefore, we can claim the second-order approximation herein: for $m>4/\alpha$,
\begin{equation}\label{eff2}
\E[N_2-n_2^*] = \frac{1}{2}-\frac{2}{\alpha}-\frac{4}{\alpha}\sum_{n=1}^{\infty}\frac{1}{n}\E\left( \sum_{i=1}^{n}W_i-\frac{5}{4}\alpha n \right)^++o(1)
\end{equation} 
as $b \to 0$, where $n_2^*=4\alpha\beta^4/b$ is the corresponding optimal sample size. 

To conduct Monte Carlo simulations, we generated pseudorandom samples from a gamma population with $\alpha=2$ and $\beta=1$. In this situation, a numerical approximation of the second-order expansion in \eqref{eff2} yields $\E[N_2-n_2^*] \approx -3.5258$. Again, we fixed the pilot sample size $m=20$, and considered a wide range of values of $b$ including $0.160,0.080,0.040,0.016,0.008$ and $0.004$ so that $n^*_2$ turned out 50, 100 (small), 200, 500 (moderate), and 1000, 2000 (large) accordingly. After running the procedure 10,000 times independently, we summarize the results in Table \ref{Tab. 2}.    

\begin{table}[h!] 
\small
\captionsetup{font=small}
\caption{Simulations from $\Gamma(2,1)$ with $m=20$ under $10,000$ runs implementing the sequential estimation procedure $\mathcal{P}_2$}
\label{Tab. 2}\par
\vskip .2cm
\centerline{\tabcolsep=3truept\begin{tabular}{cccccccc} \hline 
$b$ & $n_2^*$ & $\bar{n}_2$ & $s(n_2)$ & $\bar{n}_2/n_2^*$ & $\bar{n}_2-n_2^*$ & $\V(g(\hat{\beta}_{n_2}))$ & $\V(g(\hat{\beta}_{n_2}))/b$ \\ \hline
\multicolumn{1}{r}{$0.160$} & \multicolumn{1}{r}{$50$} & \multicolumn{1}{r}{$46.9295$} & \multicolumn{1}{r}{$18.8175$} & \multicolumn{1}{c}{$0.9386$} & \multicolumn{1}{c}{$-3.0705$} & \multicolumn{1}{c}{$0.178606$} & \multicolumn{1}{c}{$1.116289$}\\
\multicolumn{1}{r}{$0.080$} & \multicolumn{1}{r}{$100$} & \multicolumn{1}{r}{$95.2952$} & \multicolumn{1}{r}{$30.3808$} & \multicolumn{1}{c}{$0.9530$} & \multicolumn{1}{c}{$-4.7048$} & \multicolumn{1}{c}{$0.114880$} & \multicolumn{1}{c}{$1.435996$}\\
\multicolumn{1}{r}{$0.040$} & \multicolumn{1}{r}{$200$} & \multicolumn{1}{r}{$195.9255$} & \multicolumn{1}{r}{$42.3965$} & \multicolumn{1}{c}{$0.9796$} & \multicolumn{1}{c}{$-4.0745$} & \multicolumn{1}{c}{$0.051403$} & \multicolumn{1}{c}{$1.285075$}\\
\multicolumn{1}{r}{$0.016$} & \multicolumn{1}{r}{$500$} & \multicolumn{1}{r}{$496.2600$} & \multicolumn{1}{r}{$63.8226$} & \multicolumn{1}{c}{$0.9925$} & \multicolumn{1}{c}{$-3.7400$} & \multicolumn{1}{c}{$0.016814$} & \multicolumn{1}{c}{$1.050885$}\\
\multicolumn{1}{r}{$0.008$} & \multicolumn{1}{r}{$1000$} & \multicolumn{1}{r}{$995.5808$} & \multicolumn{1}{r}{$90.1453$} & \multicolumn{1}{c}{$0.9956$} & \multicolumn{1}{c}{$-4.4192$} & \multicolumn{1}{c}{$0.008248$} & \multicolumn{1}{c}{$1.031040$}\\
\multicolumn{1}{r}{$0.004$} & \multicolumn{1}{r}{$2000$} & \multicolumn{1}{r}{$1996.9680$} & \multicolumn{1}{r}{$126.3538$} & \multicolumn{1}{c}{$0.9985$} & \multicolumn{1}{c}{$-3.0320$} & \multicolumn{1}{c}{$0.004017$} & \multicolumn{1}{c}{$1.004328$}\\
\hline
\end{tabular}}
\end{table}

\subsection{Illustration 3: the rate parameter}

To estimate the rate parameter of a gamma distribution given as $g(\beta)=\beta^{-1}$, we propose the following sequential procedure $\mathcal{P}_3$:
\begin{equation}\label{P3}
N_3 = \inf\left\{ n \ge m: n \ge (\alpha b)^{-1}\hat{\beta}_n^{-2} \right\}.
\end{equation}

As was pointed out, $d=1,a_1=2$ and $k_1=-3$ so that any pilot sample size $m>6/\alpha$ will satisfy the condition \eqref{momcond}. In this case, the optimal sample size is $n^*_3=(\alpha b)^{-1}\beta^{-2}$. However, the stopping rule \eqref{P3} can hardly be rewritten following that of \cite{Woodroofe (1977)} due to the negative moment of $\hat{\beta}_n$. As a consequence, we only carry out an empirical second-order analysis via simulations here. 

Monte Carlo simulations have been conducted to demonstrate the remarkable performance of the sequential estimation procedure $\mathcal{P}_3$. We generated pseudorandom observations from a gamma distribution with $\alpha=2$ and $\beta=1$. We set the pilot sample size $m=20$ so that it was large enough. We took into consideration the values of $b$ including $0.010$, $0.005$, $0.0025$, $0.0010$, $0.0005$, and $0.00025$, respectively, so that again the optimal sample size $n^*_3$ could be determined as 50, 100 (small), 200, 500 (moderate), and 1000, 2000 (large) accordingly. The summaries from 10,000 independent runs are recorded in Table \ref{Tab. 3}.

\begin{table}[h!] 
\small
\captionsetup{font=small}
\caption{Simulations from $\Gamma(2,1)$ with $m=20$ under $10,000$ runs implementing the sequential estimation procedure $\mathcal{P}_3$}
\label{Tab. 3}\par
\vskip .2cm
\centerline{\tabcolsep=3truept\begin{tabular}{cccccccc} \hline 
$b$ & $n_3^*$ & $\bar{n}_3$ & $s(n_3)$ & $\bar{n}_3/n_3^*$ & $\bar{n}_3-n_3^*$ & $\V(g(\hat{\beta}_{n_3}))$ & $\V(g(\hat{\beta}_{n_3}))/b$ \\ \hline
\multicolumn{1}{r}{$0.010$} & \multicolumn{1}{r}{$50$} & \multicolumn{1}{r}{$50.8874$} & \multicolumn{1}{r}{$9.9734$} & \multicolumn{1}{c}{$1.0177$} & \multicolumn{1}{c}{$0.8874$} & \multicolumn{1}{c}{$0.010076$} & \multicolumn{1}{c}{$1.007606$}\\
\multicolumn{1}{r}{$0.005$} & \multicolumn{1}{r}{$100$} & \multicolumn{1}{r}{$100.9992$} & \multicolumn{1}{r}{$14.0009$} & \multicolumn{1}{c}{$1.0010$} & \multicolumn{1}{c}{$0.9992$} & \multicolumn{1}{c}{$0.004965$} & \multicolumn{1}{c}{$0.992902$}\\
\multicolumn{1}{r}{$0.0025$} & \multicolumn{1}{r}{$200$} & \multicolumn{1}{r}{$200.0173$} & \multicolumn{1}{r}{$19.8177$} & \multicolumn{1}{c}{$1.0051$} & \multicolumn{1}{c}{$1.0173$} & \multicolumn{1}{c}{$0.002467$} & \multicolumn{1}{c}{$0.986840$}\\
\multicolumn{1}{r}{$0.0010$} & \multicolumn{1}{r}{$500$} & \multicolumn{1}{r}{$500.9509$} & \multicolumn{1}{r}{$31.2984$} & \multicolumn{1}{c}{$1.0019$} & \multicolumn{1}{c}{$0.9509$} & \multicolumn{1}{c}{$0.000982$} & \multicolumn{1}{c}{$0.981624$}\\
\multicolumn{1}{r}{$0.0005$} & \multicolumn{1}{r}{$1000$} & \multicolumn{1}{r}{$1001.2121$} & \multicolumn{1}{r}{$44.5421$} & \multicolumn{1}{c}{$1.0012$} & \multicolumn{1}{c}{$1.2121$} & \multicolumn{1}{c}{$0.000497$} & \multicolumn{1}{c}{$0.993567$}\\
\multicolumn{1}{r}{$0.00025$} & \multicolumn{1}{r}{$2000$} & \multicolumn{1}{r}{$2000.2164$} & \multicolumn{1}{r}{$63.4808$} & \multicolumn{1}{c}{$1.0001$} & \multicolumn{1}{c}{$0.2164$} & \multicolumn{1}{c}{$0.000252$} & \multicolumn{1}{c}{$1.007719$}\\
\hline
\end{tabular}}
\end{table}

\subsection{Illustration 4: the survival probability}

One may also be interested in survival probability as a function of $\beta$, say $g(\beta)=\Pr(X>c)$, where $c$ is a positive constant. The corresponding sequential estimation procedure $\mathcal{P}_4$ can be then constructed by
\begin{equation}\label{P4}
N_4 = \inf\left\{ n \ge m: n \ge \frac{c^{2\alpha}\hat{\beta}_n^{-2\alpha}e^{-2c/\hat{\beta}_n}}{\{\Gamma(\alpha)\}^2\alpha b} \right\}.
\end{equation}

As was pointed out, $d=2$ and $\min\{k_1,k_2\}=-(\alpha+3)$ so that any pilot sample size $m>2+6/\alpha$ will satisfy the condition \eqref{momcond}. In this scenario, the optimal sample size is $n^*_4=\frac{c^{2\alpha}\beta^{-2\alpha}e^{-2c/\beta}}{\{\Gamma(\alpha)\}^2\alpha b}$. Similarly, we are not going to do a second-order analysis theoretically as the stopping rule \eqref{P4} can hardly be rewritten in the spirit of \cite{Woodroofe (1977)}. Instead, we investigate the second-order efficiency via simulations. 

For Monte Carlo simulations, we generated pseudorandom observations from a gamma distribution with $\alpha=2$ and $\beta=2$ to estimate the probability $\Pr(X>3)$, that is, the constant $c=3$. A pilot sample size of $m=20$ would be large enough in this situation. To make the optimal sample size $n_4^*$ range among the values including 50, 100 (small), 200, 500 (moderate), and 1000, 2000 (large), we computed the corresponding values of $b$ in turn. That is, $b=0.00252,0.00126,0.00063,0.000252,0.000126$ and $0.000063$. The simulated findings from 10,000 independent runs are summarized in Table \ref{Tab. 4}.

\begin{table}[h!] 
\small
\captionsetup{font=small}
\caption{Simulations from $\Gamma(2,2)$ with $c=3$ and $m=20$ under $10,000$ runs implementing the sequential estimation procedure $\mathcal{P}_4$}
\label{Tab. 4}\par
\vskip .2cm
\centerline{\tabcolsep=3truept\begin{tabular}{cccccccc} \hline 
$b$ & $n_4^*$ & $\bar{n}_4$ & $s(n_4)$ & $\bar{n}_4/n_4^*$ & $\bar{n}_4-n_4^*$ & $\V(g(\hat{\beta}_{n_4}))$ & $\V(g(\hat{\beta}_{n_4}))/b$ \\ \hline
\multicolumn{1}{r}{$0.00252$} & \multicolumn{1}{r}{$50$} & \multicolumn{1}{r}{$50.0294$} & \multicolumn{1}{r}{$4.9539$} & \multicolumn{1}{c}{$1.0006$} & \multicolumn{1}{r}{$0.0294$} & \multicolumn{1}{c}{$0.002606$} & \multicolumn{1}{c}{$1.033928$}\\
\multicolumn{1}{r}{$0.00126$} & \multicolumn{1}{r}{$100$} & \multicolumn{1}{r}{$100.0929$} & \multicolumn{1}{r}{$7.0023$} & \multicolumn{1}{c}{$1.0009$} & \multicolumn{1}{r}{$0.0929$} & \multicolumn{1}{c}{$0.001265$} & \multicolumn{1}{c}{$1.003583$}\\
\multicolumn{1}{r}{$0.00063$} & \multicolumn{1}{r}{$200$} & \multicolumn{1}{r}{$200.1026$} & \multicolumn{1}{r}{$9.9768$} & \multicolumn{1}{c}{$1.0005$} & \multicolumn{1}{r}{$0.1026$} & \multicolumn{1}{c}{$0.000635$} & \multicolumn{1}{c}{$1.007231$}\\
\multicolumn{1}{r}{$0.000252$} & \multicolumn{1}{r}{$500$} & \multicolumn{1}{r}{$500.0495$} & \multicolumn{1}{r}{$15.7865$} & \multicolumn{1}{c}{$1.0001$} & \multicolumn{1}{r}{$0.0495$} & \multicolumn{1}{c}{$0.000252$} & \multicolumn{1}{c}{$1.000772$}\\
\multicolumn{1}{r}{$0.000126$} & \multicolumn{1}{r}{$1000$} & \multicolumn{1}{r}{$1000.1743$} & \multicolumn{1}{r}{$22.3803$} & \multicolumn{1}{c}{$1.0002$} & \multicolumn{1}{r}{$0.1743$} & \multicolumn{1}{c}{$0.000127$} & \multicolumn{1}{c}{$1.006751$}\\
\multicolumn{1}{r}{$0.000063$} & \multicolumn{1}{r}{$2000$} & \multicolumn{1}{r}{$1999.7089$} & \multicolumn{1}{r}{$31.6515$} & \multicolumn{1}{c}{$0.9999$} & \multicolumn{1}{r}{$-0.2911$} & \multicolumn{1}{c}{$0.000063$} & \multicolumn{1}{c}{$1.003911$}\\
\hline
\end{tabular}}
\end{table}

\subsection{Some overall sentiments}
From Tables \ref{Tab. 1}-\ref{Tab. 4} and many other simulations (with varying values of $\alpha,\beta,$ and $m$) that we have run, we may summarize our overall sentiments as follows.
\begin{enumerate}
    \item[(a)] Different choices of the underlying gamma population distribution and the pilot sample size $m$ have not impacted the overall performances of our proposed sequential estimation procedure \eqref{SR} in a significant way. 
    \item[(b)] For all of those four illustrations under consideration, we observe that both $\bar{n}_i/n^*_i$ and $\V(g(\hat{\beta}_{n_i}))/b$ are close to 1 across the board, $i=1,2,3,4$. It shows that our proposed sequential estimation procedure \eqref{SR} performs impressively.
    \item[(c)] For Illustrations 1 and 2, we have carried out second-order analyses by applying \citeauthor{Woodroofe (1977)}'s (\citeyear{Woodroofe (1977)}) nonlinear renewal theory. The simulations have empirically validated our theoretical results that as $b\to0$,
    $$\E[N_1-n^*_1]=-0.5680+o(1), \text{ and } \E[N_2-n^*_2]=-3.5258+o(1).$$
    \item[(d)] For Illustrations 3 and 4, \citeauthor{Woodroofe (1977)}'s (\citeyear{Woodroofe (1977)}) nonlinear renewal theory fails to work. However, the simulations have suggested that both $\lim_{b\to0}\E[N_3-n^*_3]$ and $\lim_{b\to0}\E[N_4-n^*_4]$ are bounded. To be specific, as $b\to0$,
    $$\E[N_3-n^*_3] \approx 1, \text{ and } \E[N_4-n^*_4] \approx 0.$$
    \item[(e)] For Illustrations 1-4, we have the following expressions of the optimal sample sizes as per \eqref{oss}:
    $$n_i^* = k_ib^{-1}, ~ i = 1,2,3,4,$$
    where $k_1=\alpha\beta^2$, $k_2=4\alpha\beta^4$, $k_3=\alpha^{-1}\beta^{-2}$, and $k_4=\{\Gamma(\alpha)\}^{-2}c^{2\alpha}\beta^{-2\alpha}e^{-2c/\beta}$. Consequently, the magnitudes of $n_i^*$'s depend on $\alpha,\beta,$ and $c$, and we cannot claim that one is uniformly larger or smaller than another. For example, when $\alpha=1/4$ and $\beta=2$, $n^*_2<n^*_3$; but when $\alpha=2$ and $\beta=2$, $n^*_2>n^*_3$. It indicates that for a fixed level $b$, even if we sample observations from the same gamma population, the sample sizes needed under various illustrations can be quite different. Analogously, for comparable optimal sample sizes, different levels of $b$ will be achieved under various illustrations, as is reflected in Tables \ref{Tab. 1}-\ref{Tab. 4}.
\end{enumerate}

\setcounter{equation}{0}
\section{Real data analysis}\label{RDA}

In this section, we implement the developed sequential techniques in real-world studies using two separate real data sets: (i) survival times data from a group of 97 Swiss female dementia patients diagnosed at age 70–74, referred to as the \textit{dementia data}; and (ii) the number of days that the seeds of marigold need to flower, referred to as the \textit{marigold data}.

\subsection{Analysis using dementia data}\label{Sect4.1}

The survival times (in years) after diagnosis of dementia in a group of 97 Swiss females with age at diagnosis 70-74 inclusive were given in Example 5.1 of \cite{Elandt-Johnson and Johnson (1999)}. We treat the dementia data as our population of size $n=97$, to which the exponential distribution provides a reasonably good fit as the Kolmogorov-Smirnov test returns a p-value of 0.2384. That is, the shape parameter $\alpha$ is known to be 1, but the scale parameter $\beta$ is assumed unknown. Recently, \cite{Zhuang et al. (2020)} studied this data set and constructed a fixed-accuracy confidence interval for the survival probability $\Pr(X>4.6)$, where $X$ denotes the survival time of a female dementia patient, and 4.6 years is the estimated median survival time from onset of dementia to death for women according to an analysis conducted by \cite{Xie et al. (2008)}. We denote the 97 survival times by $X_1,X_2,...,X_{97}$, and summarize the following parameters along with their definitions in Table \ref{Tab. 5}. 

\begin{table}[h!] 
\small
\captionsetup{font=small}
\caption{Dementia data parameters}
\label{Tab. 5}\par
\vskip .2cm
\centering
\begin{tabular}{lcccc}
\hline
Parameter & Mean $\mu$ & Variance $\sigma^2$ & Rate & $\Pr(X>4.6)$\\
\hline
Definition & $n^{-1}\sum_{i=1}^{n}X_i$ & $n^{-1}\sum_{i=1}^{n}(X_i-\mu)^2$ & $\mu^{-1}$ & $\sum_{i=1}^{n}I(X_i>4.6)$ \\
\hline
Value & 5.180 & 20.249 & 0.193 & 0.439\\
\hline
\end{tabular}
\end{table}

To implement the sequential estimation procedures \eqref{P1}, \eqref{P2}, \eqref{P3}, and \eqref{P4}, we first collect a pilot sample of size $m=20$, based on which we obtain preliminary estimates of the mean, the variance, the rate, and the survival probability $\Pr(X>4.6)$ to determine a reasonable level of $b$ for each procedure. Then, we sample one observation at a time successively until termination. The results are reported in Table \ref{Tab. 6}. We find that while all these four procedures lead to substantially smaller terminal sample sizes, the estimates are comparable with the population parameters.

\begin{table}[h!] 
\small
\captionsetup{font=small}
\caption{Analysis using the dementia data}
\label{Tab. 6}\par
\vskip .2cm
\centering
\begin{tabular}{lcccc}
\hline
Parameter & Procedure & $b$ & Terminal sample size & Estimate \\
\hline
Mean & \eqref{P1} & 0.5 & 50 & 4.960\\
Variance & \eqref{P2} & 50 & 39 & 22.020\\
Rate & \eqref{P3} & 0.001 & 42 & 0.204\\
$\Pr(X>4.6)$ & \eqref{P4} & 0.003 & 45 & 0.444\\
\hline
\end{tabular}
\end{table}

\subsection{Analysis using marigold data}\label{Sect4.2}

In a lot of cases, the underlying population is better approximated by a normal distribution rather than by a gamma distribution, and the interest lies in the estimation of the population variance with some desired precision. Based on the normal-gamma transformation proposed in \cite{Zacks and Khan (2011)}, if $X_1,...,X_n$ are i.i.d. normal random variables with mean $\mu$ and variance $\sigma^2$, then 
\begin{equation}\label{ngtrans}
Y_i=\frac{i}{i+1}(X_{i+1}-\bar{X}_{i})^2, \ 1 \le i \le n-1    
\end{equation} 
are independent random variables from a gamma distribution with a known shape $\alpha = 1/2$ and an unknown scale $\beta = 2\sigma^2$. In this sense, estimating the unknown normal variance $\sigma^2$ using $X_1,...,X_n$ is equivalent to estimating the gamma mean using $Y_1,...,Y_{n-1}$. Therefore, our proposed sequential estimation procedure \eqref{P1} helps to provide a satisfactory solution to this inference problem. 

For illustrative purposes, we analyze the marigold data of size $n=460$ collected by \cite{Mukhopadhyay et al. (2004)} on the number of days that the seeds of marigold variety 2 needed to flower. We treat this real data set as our population, which does not seem to contradict normality as is confirmed via the Shapiro-Wilk test with an associated p-value of 0.1473. The population variance is also computed for reference, which turns out 16.226.

We performed the Zacks and Khan's (\citeyear{Zacks and Khan (2011)}) transformation \eqref{ngtrans} onto the marigold data to get a random sample of 459 observations from a gamma population with $\alpha=1/2$ and some unknown $\beta$. From here, we assign $m=20$ and $b=2$, and proceeded the sequential estimation procedure \eqref{P1} to estimate the mean. The result is summarized in Table  \ref{Tab. 7}. Similarly, we find that the procedure leads to a substantially smaller sample size, but the estimate is close to the corresponding population parameter.

\begin{table}[h!] 
\small
\captionsetup{font=small}
\caption{Analysis using the marigold data}
\label{Tab. 7}\par
\vskip .2cm
\centering
\begin{tabular}{lcccc}
\hline
Parameter & Procedure & $b$ & Terminal sample size & Estimate \\
\hline
Mean & \eqref{P1} & 2 & 291 & 17.002\\
\hline
\end{tabular}
\end{table}

\section{Conclusions} \label{Conc}

In this paper, we have proposed a sequential sampling procedure for estimating $g(\beta)$, a general function of the scale parameter $\beta$ in a gamma distribution, with its variance bounded by a prefixed small level, given that the shape parameter $\alpha$ is known. Our work has enriched the existing literature in the following two aspects. 

First, it has extended the research of \cite{Mahmoudi et al. (2019)} which merely focused on the bounded risk point estimation of the scale parameter $\beta$ itself in a gamma distribution. Estimating a function $g(\beta)$ offers greater generality, and can better serve statistical inference problems where parameters other than $\beta$ are of interest, as is illustrated in Section \ref{Illu}.

Second, we have relaxed the conditions imposed on the function $g(\cdot)$ compared with \cite{Mukhopadhyay and Wang (2019)} and \cite{Mukhopadhyay (2021)}, so that more functions can be incorporated into the general framework we have formulated, such as the rate parameter $g(\beta)=\beta^{-1}$ and the survival probability $g(\beta)=\Pr(X>c)$. This helps to make our sequential estimation procedure more flexible in practical applications. Although the condition \eqref{C1} is not difficult to meet, it is still a sufficient condition alone for Theorems \ref{Thm1}-\ref{Thm2} to hold. There is a possibility to further loosen this condition so that our theoretical framework can be applied to an even wider range of functions and problems. We will follow this direction to conduct more research in the future.    

We have followed up with four interesting illustrations with the function $g(\beta)$ substituted with the gamma mean, the gamma variance, the gamma rate parameter, and a survival probability, respectively. Appealing first-order and/or second-order properties are exhibited. Both Monte Carlo simulations and real data analysis are conducted to demonstrate the remarkable performance of the procedures. In particular, Section \ref{Sect4.2} addresses an unusual application of our proposed sequential estimation procedure \eqref{P1}. Although the theoretical framework is established under an underlying gamma population distribution, it can be implemented to estimate an unknown normal variance with the variance of its estimator bounded by a prefixed small level.

\section*{Acknowledgements}

We would like to express our gratitude to the Editor-in-Chief, the Associate Editor, and the four reviewers for their invaluable contributions to the improvement of our manuscript. We are particularly grateful to the Associate Editor, whose recommendation of the relevant paper, \cite{Mukhopadhyay and Li (2024)}, is helpful in broadening our references.

\end{document}